\documentclass[sn-mathphys-num]{sn-jnl}

\usepackage{graphicx}%
\usepackage{multirow}%
\usepackage{amsmath,amssymb,amsfonts}%
\usepackage{amsthm}%
\usepackage{mathrsfs}%
\usepackage[title]{appendix}%
\usepackage{xcolor}%
\usepackage{textcomp}%
\usepackage{manyfoot}%
\usepackage{booktabs}%
\usepackage{algorithm}%
\usepackage{algorithmicx}%
\usepackage{algpseudocode}%
\usepackage{listings}%
\usepackage{newtxtext,newtxmath}

\theoremstyle{thmstyleone}%
\newtheorem{Thm}{Theorem}
\newtheorem{Lem}{Lemma}
\newtheorem{Cor}{Corollary}

\theoremstyle{thmstyletwo}%

\theoremstyle{thmstylethree}%

\begin{document}

\title[Necessary and sufficient condition for hysteresis in the mathematical model of the cell type regulation of \textit{Bacillus subtilis}]{Necessary and sufficient condition for hysteresis in the mathematical model of the cell type regulation of \textit{Bacillus subtilis}}


\author*[1]{\fnm{Sohei} \sur{Tasaki}}\email{tasaki@math.sci.hokudai.ac.jp}

\author[2]{\fnm{Madoka} \sur{Nakayama}}\email{nakayama.madoka@tmd.ac.jp}

\author[3]{\fnm{Izumi} \sur{Takagi}}\email{i.takagi@tohoku.ac.jp}

\author[4]{\fnm{Jun-ichi} \sur{Wakita}}\email{wakita@phys.chuo-u.ac.jp}

\author[5]{\fnm{Wataru} \sur{Shoji}}\email{wshoji@gmail.com}

\affil*[1]{\orgdiv{Department of Mathematics, Faculty of Science}, \orgname{Hokkaido University}, \orgaddress{\street{Kita 10, Nishi 8, Kita-ku}, \city{Sapporo}, \postcode{0600810}, \state{Hokkaido}, \country{Japan}}}

\affil[2]{\orgdiv{Institute for Liberal Arts}, \orgname{Institute of Science Tokyo}, \orgaddress{\street{2-8-1 Konodai}, \city{Ichikawa}, \postcode{2720827}, \state{Chiba}, \country{Japan}}}

\affil[3]{\orgdiv{Mathematical Institute}, \orgname{Tohoku University}, \orgaddress{\street{6-3 Aramaki-aza-Aoba, Aoba-ku}, \city{Sendai}, \postcode{9808579}, \state{Miyagi}, \country{Japan}}}

\affil[4]{\orgdiv{Department of Physics}, \orgname{Chuo University}, \orgaddress{\street{1-13-27 Kasuga}, \city{Bunkyo-ku}, \postcode{1128551}, \state{Tokyo}, \country{Japan}}}

\affil[5]{\orgdiv{Frontier Research Institute for Interdisciplinary Sciences (FRIS)}, \orgname{Tohoku University}, \orgaddress{\street{6-3 Aramaki-aza-Aoba, Aoba-ku}, \city{Sendai}, \postcode{9808579}, \state{Miyagi}, \country{Japan}}}


\abstract{The key to a robust life system is to ensure that each cell population is maintained in an appropriate state. 
In this work, a mathematical model is used to investigate the control of the switching between the migrating and non-migrating states of the \textit{Bacillus subtilis} cell population. 
In this case, the motile cells and matrix producers are the predominant cell types in the migrating cell population and non-migrating state, respectively, and can be suitably controlled according to the environmental conditions and cell density information. 
A minimal smooth model consisting of four ordinary differential equations is used as the mathematical model to control the \textit{B. subtilis} cell types. 
Furthermore, the necessary and sufficient conditions for the hysteresis, which pertains to the change in the pheromone concentration, are clarified. 
In general, the hysteretic control of the cell state enables stable switching between the migrating and growth states of the \textit{B. subtilis} cell population, thereby facilitating the biofilm life cycle. 
The results of corresponding culture experiments are examined, and the obtained corollaries are used to develop a model to input environmental conditions, especially, the external pH. 
On this basis, the environmental conditions are incorporated in a simulation model for the cell type control. In combination with a mathematical model of the cell population dynamics, a prediction model for colony growth involving multiple cell states,  including concentric circular colonies of \textit{B. subtilis}, can be established.}




\maketitle

\section{Introduction}
\label{sec:intro}
Cellular state diversity is the source of the morphology and function of life systems. Even in prokaryotes, the robust structure of bacterial biofilms can be attributed to the heterogeneous presence of various types of cells \citep{RN367,RN369,RN229,RN378,RN342,RN323,RN341,RN373}. Among such cells, \textit{Bacillus subtilis} is considered to be the master of differentiation, as it can exhibit extremely diverse cell types \citep{RN106}. The cellular diversity helps develop a diverse colony morphology \citep{RN238} and complex biofilm structure, supporting long-term survival and growth in response to environmental variations \citep{RN355}. In this study, we consider the following mathematical model that describes the cell type regulation of \textit{B. subtilis} \citep{RN389}:
\begin{equation}
\label{eq:SHAC}
    \left\{
    \begin{alignedat}{2}
	\frac{dS}{dt} &= \frac{c_S H}{a_S + b_S C} - d_S S \\
	\frac{dH}{dt} &= \frac{c_H}{a_H + b_H A} - d_H H \\
        \frac{dA}{dt} &= \frac{c_A}{a_A + b_A S} - d_A A  \\
        \frac{dC}{dt} &= c_C X A - d_C C 
    \end{alignedat}
    \right.
\end{equation}
In this model, the cell states are described by four variables $S=S(t)$, $H=H(t)$, $A=A(t)$ and $C=C(t)$, each of which represents a group of cooperating genes and their products (Fig.~\ref{fig:1}A). Specifically,
$S$ is a group represented by Spo0A$\sim$P (including the phosphorelay of Spo0F, Spo0B, and Spo0A); $H$ and $A$ correspond to SigH and AbrB, respectively; and $C$ is a group represented by ComK (driven by the ComX-ComP-ComA pathway). 

Here, $a_{\ast}$,  $b_{\ast}$, $c_{\ast}$, $d_{\ast}$ $(\ast = S, H, A, C )$ are all positive constants; $a_S$ represents the baseline inhibition rate of $S$, $b_S$ represents the inhibition rate of $S$ by $C$, and $c_S$ represents the activation rate of $S$ by $H$. $a_H$ represents the baseline inhibition rate of $H$, $b_H$ represents the inhibition rate of $H$ by $A$, and $c_H$ represents the baseline activation rate of $H$. 
$a_A$ represents the baseline inhibition rate of $A$, $b_A$ represents the inhibition rate of $A$ by $S$, and $c_A$ represents the baseline activation rate of $A$.
$c_C$ represents the expression rate of $C$ activated by $A$ (and $X$, explained later). 
$d_S$, $d_H$, $d_A$, and $d_C$ represent the decay rates of $S$, $H$, $A$, and $C$, respectively.

\begin{figure*}[tbp]
\begin{center}
\includegraphics[width=135mm]{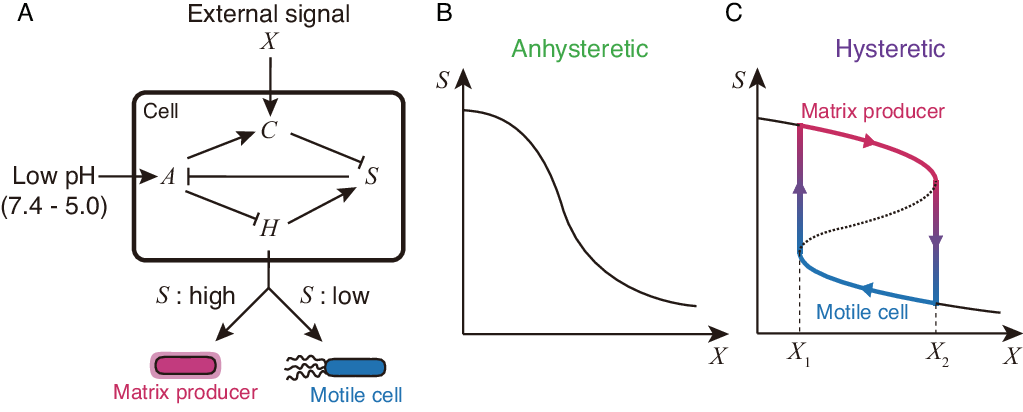}
\end{center}
\caption{Model to determine the response of the cells and cell populations to the environmental pH. 
(A) Model for the cell type selection in response to the environmental pH and cell density. 
(B and C) Two types of cell type controls. Anhysteretic (B) and hysteretic (C).}
\label{fig:1}
\end{figure*}

The output of this system is the cell type, which corresponds to a matrix producer and motile cell when $S$ is high and low, respectively (Fig.~\ref{fig:1}A). Sporulation is initiated when $S$ exhibits continuously high values; however, this aspect is not considered in this work. Instead, this work is focused on examining the switch between the migrating (planktonic) and non-migrating (biofilm) states \citep{RN263,RN343,RN327,RN106,RN325}. The inputs of the cell state control system include the external environmental conditions and auto-inducing signals that represent the cell density. Among such signals, one corresponds to a small peptide ComX secreted by \textit{B. subtilis} cells. In the following text, the concentration of this peptide is denoted as $X$, and it is an external parameter input from $C$ to the cell type regulation system \eqref{eq:SHAC}. 

When cells are motile, they disperse spatially, reducing cell density and causing $X$ to decrease. 
When cells proliferate without moving, cell density increases, causing $X$ to rise. Additionally, $X$ diffuses through space. Mathematical modeling of such dynamics of $X$ and coupling them effectively with the current cell type control model could yield a mathematical model (possibly described by PDEs) for colony pattern formation involving different cell types, though this topic exceeds the scope of this paper. 
Here, two biological hierarchies exist: the cellular level (gene regulation level) and the tissue level (colony formation level). Integrating phenomena at different hierarchical levels requires novel and technically challenging approaches.

The curve of the set of equilibrium points of \eqref{eq:SHAC} can be divided into two types pertaining to the increase and decrease in $X$, indicating that the cell state (for example, $S$) is monotonic and non-monotonic (Figs.~\ref{fig:1}B and C), respectively. 
In the latter case, the curve is an S-shaped curve with two turning points. Specifically, there exist two cases in which the cell type control is not hysteretic (Fig.~\ref{fig:1}B) and hysteretic (Fig.~\ref{fig:1}C), which pertains to the increase/decrease in the cell density signal $X$, respectively. In general, the hysteretic control is necessary to facilitate the biofilm life cycle or concentric colony formation. When there is hysteresis regarding the increase or decrease of $X$ with respect to the choice of cell state, a life cycle for the cell population occurs as follows: (Migration phase) The cell density information $X$ at the growth front of the cell population decreases as it disperses in the motility state, and the state switches to the matrix-production state when the concentration falls below a certain threshold $X_1$. (Growth phase) As the cell population grows and matures in the matrix production state, $X$ increases, and when it exceeds a certain threshold $X_2 \left( >X_1 \right)$, the state switches to the motility state. 
In this regard, the objective of this study is to classify the presence or absence of hysteresis in the selection of such cell states through certain parameters.
Furthermore, the mechanism to control the state of the cells and cell populations in response to environmental conditions is discussed.

\section{Results}
\label{sec:results}
\subsection{Necessary and sufficient condition for hysteretic cell type regulation}
\label{subsec:thm}
The main theorem described herein is a mathematical and formal claim. In this context, the meaning of the parameters may be difficult to understand. This type of unbiased form of writing is intended to facilitate the subsequent testing of two different interpretations.

To describe the results, first, the steady state hysteresis is defined. 
A set of steady states is considered to be anhysteretic if the steady state is unique to $X$, and the steady state $S$ decreases monotonically with respect to $X$ (Fig.~\ref{fig:1}B). 
In contrast, a set of steady states is considered to be hysteretic if the steady state is not unique to $X$, and the steady state $S$ is an (inverse) S-shaped curve (Fig.~\ref{fig:1}C). 
When the steady state is anhysteretic, any equilibrium point on the curve $X=X(S)$ of the set of steady states is stable. 
In comparison, when the steady state is hysteretic, any equilibrium point between the two folding points ($X'(S)=0$) is unstable, and the outer equilibrium point is stable (\ref{sec:stability}).

The parameters for the classification can be defined as follows: 
\begin{equation}
\label{eq:parameters0}
\begin{aligned}
D &= a_A - \frac{b_H c_A}{a_H d_A}, \quad 
\hat{P} = \frac{a_H b_H a_A c_A d_A}{\left( a_H a_A d_A - b_H c_A \right)^2}, \\
\hat{Q} &= \frac{1}{2} \left\{ \frac{c_S c_H b_A d_A}{a_S d_S d_H \left( b_H c_A - a_H a_A d_A \right)} - 1 \right\}. 
\end{aligned}
\end{equation}
In this case, the following holds true. 

\begin{Thm}
\label{thm:1}
If $D \geq 0$, the set of steady states is anhysteretic (Fig.~\ref{fig:1}B). 
If $D < 0$, 
any $\hat{P}$ admits $\hat{Q}_{\rm{c}}(\hat{P})$ such that 
\begin{enumerate}
\item If $\hat{Q} \leq \hat{Q}_{\rm{c}}(\hat{P})$, the set of steady states is anhysteretic (Fig.~\ref{fig:1}B). 
\item If $\hat{Q} > \hat{Q}_{\rm{c}}(\hat{P})$, the set of steady states is hysteretic (Fig.~\ref{fig:1}C). 
\end{enumerate}
Moreover, the threshold $\hat{Q}_{\rm{c}} = \hat{Q}_{\rm{c}}(\hat{P})$ satisfies the following (Fig.~\ref{fig:2}). 
\begin{equation}
\label{eq:qcprofile}
\hat{Q}_{\rm{c}}'(\hat{P}) >0, \quad 
\hat{Q}_{\rm{c}}''(\hat{P}) >0, \quad 
	\lim_{\hat{P}\rightarrow +0} \frac{\hat{Q}_{\rm{c}}(\hat{P})}{\hat{P}^{2/3}} = \frac{3}{2}, \quad 
	\lim_{\hat{P}\rightarrow +\infty} \frac{\hat{Q}_{\rm{c}}(\hat{P})}{\hat{P}^{2}} = \frac{27}{2} . 
\end{equation}
\end{Thm}

\begin{figure*}[tbp]
\begin{center}
\includegraphics[width=87mm]{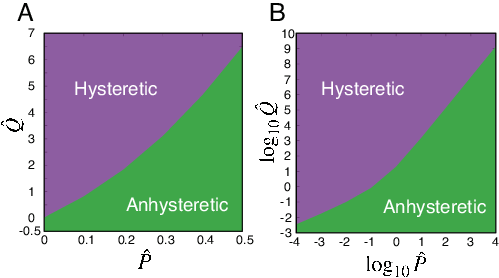}
\end{center}
\caption{Range of hysteresis in the $\hat{P}-\hat{Q}$ plane. 
(A) Normal graph. 
(B) Log-log graph.}
\label{fig:2}
\end{figure*}

\subsection{Environmental pH dependent cell type regulation}
\label{subsec:cor}
Among the structures of steady-state solutions revealed by Theorem~\ref{thm:1}, the exponents describing the asymptotic behavior of the threshold curve, in particular, are thought to reflect the nonlinearity of the model incorporating the characteristics of gene regulatory networks. Understanding this correspondence would be remarkable both mathematically and biologically, but its analysis is challenging and remains a future work. 

A focus of this study is to examine the relation between the cell type cycle generation conditions and the environmental factors,
specifically, to examine the mechanism using which each cell of \textit{B. subtilis} reflects the environmental changes in the control pattern of the cell type.
In this regard, we consider the environmental pH as a sample environmental factor. 
As mentioned previously, when the control of two cell types, motile cells and matrix producers, is hysteretic, a life cycle occurs in the cell population growth process, known as the biofilm life cycle. One of the simplest observations of the cell population life cycle is the formation of concentric colonies \citep{RN235,RN301,RN303,RN149,RN184}. 
This colony growth pattern alternates between growth and migration phases. 
The dominant cell type for colony growth periodically switches between matrix producers and motile cells. 
In other words, the formation of concentric colonies depends on the presence or absence of hysteresis in the cell type control.
In a recent study, the relationship between concentric colony formation and environmental pH was clarified \citep{RN389}. 
It was noted that under appropriate conditions, concentric colonies are formed on a solid nutrient medium containing approximately 0.7\% agar. 
Initially, in the neutral region (pH 6.8--8.0) with an intracellular pH of 7.4 \citep{RN247}, concentric colonies are not formed, and only the growth phase through the matrix production cells is observed (Fig.~\ref{fig:3}A). 
As the environmental pH decreases to approximately 6.8, many extremely short migration phases appear. When the pH is close to this transition point, there exists a considerable variation in space, the periodicity is not clear, and the pattern is considerably different from concentric circles. When the pH reduces to less than 6.5, concentric circular colonies that expand periodically are formed. When the pH is less than 5.3, colony formation becomes unstable and stops halfway, or no colony is formed. Moreover, at a pH less than 5.1, colonies are never formed. 

\begin{figure*}[tbp]
\begin{center}
\includegraphics[width=135mm]{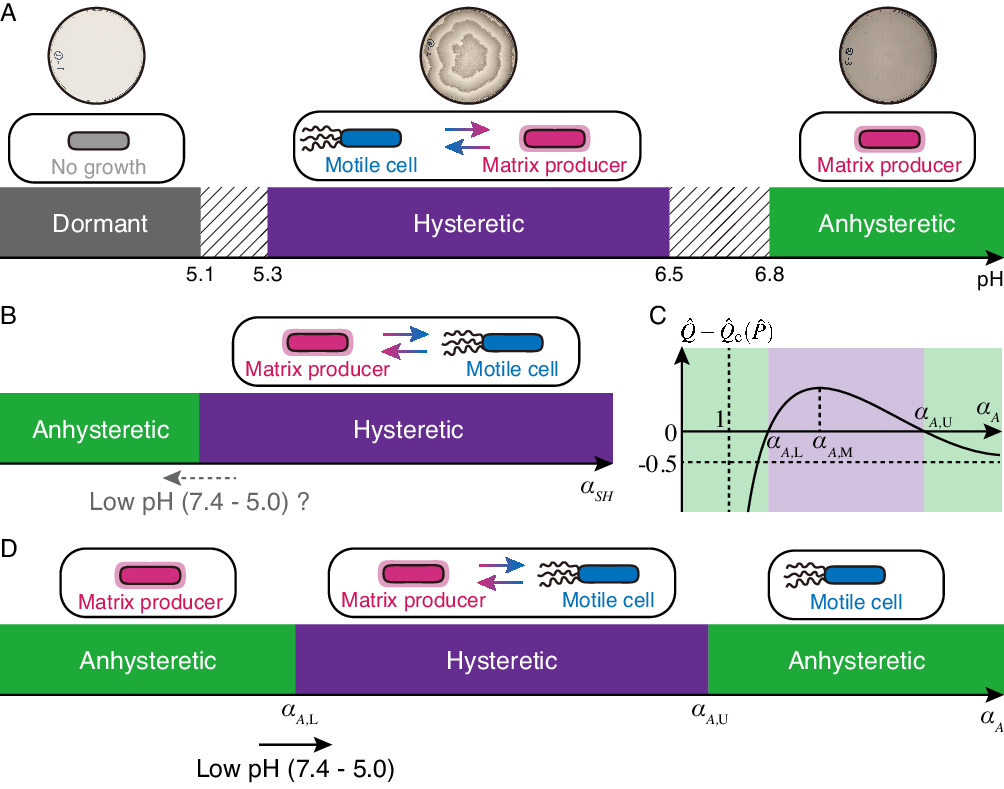}
\end{center}
\caption{Hysteresis and parameter change. 
(A) Hysteresis and environmental pH. The schematic pertains to concentric colony formation experiments \citep{RN389}. 
(B) Hysteresis and external activation of Spo0A$\sim$P and SigH, $\alpha_{SH}$. 
(C) Curve of $\hat{Q}-\hat{Q}_{\rm{c}}(\hat{P})$ vs. $\alpha_A$. 
(D) Hysteresis and external activation of AbrB, $\alpha_A$}. 
\label{fig:3}
\end{figure*}

It is crucial to determine which of the gene (group) variables $S$, $H$, $A$, and $C$ plays a role in environmental susceptibility.   To this end, we will relate the classification parameters $D$, $\hat{P}$, $\hat{Q}$, and $\hat{Q}_{c}$ to the activity levels of $S$, $H$, $A$, and $C$. 
Let $p=a_A$ and $q=\dfrac{b_H c_A}{a_H d_A}$. 
The sign of $D=p-q$ is equivalent to whether $\alpha_A = q/p$ is less than or greater than $1$. 
This parameter $\alpha_A$ can be expressed as 
\begin{alignat*}{2}
\alpha_A &=\frac{b_H}{a_H} \times \frac{c_A / a_A}{d_A} \\
 &= \text{(rate of suppression of $H$ by $A$)} \times \frac{\text{(basic production rate of $A$)}}{\text{(decay rate of $A$)}}
\end{alignat*}
which indicates the rate at which $A$, AbrB, is activated from outside the model system of the cell type regulation \eqref{eq:SHAC}. 
Similarly, we can define the rate at which $S$ and $H$, Spo0A$\sim$P and SigH, are activated from outside the model system as
\begin{alignat*}{2}
\alpha_{SH} &=\frac{b_A}{a_A} \times \frac{c_S / a_S}{d_S} \times \frac{c_H / a_H}{d_H} \\
 &= \text{(rate of suppression of $A$ by $S$)} \\
 & \quad \times \frac{\text{(basic production rate of $S$ activated by $H$)}}{\text{(decay rate of $S$)}} \\
 & \quad \times \frac{\text{(basic production rate of $H$)}}{\text{(decay rate of $H$)}} . 
\end{alignat*}
The classification parameters \eqref{eq:parameters0} can be expressed as
\begin{equation}
\hat{P} = \frac{\alpha_A}{\left( \alpha_A -1 \right)^2} , \quad 
\hat{Q} = \frac{1}{2} \left( \frac{\alpha_{SH}}{\alpha_A - 1} -1 \right) , 
\end{equation}
and the following holds from Theorem~{\ref{thm:1}}. 

\begin{Cor}
\label{cor:1}
If $\alpha_A \leq 1$, the set of steady states is anhysteretic. 
If $\alpha_A > 1$, there exists $\alpha_{SH, \rm{c}}(\alpha_A )$ such that 
\begin{enumerate}
\item If $\alpha_{SH} \leq \alpha_{SH, \rm{c}}(\alpha_A )$, the set of steady states is anhysteretic. 
\item If $\alpha_{SH} > \alpha_{SH, \rm{c}}(\alpha_A )$, the set of steady states is hysteretic. 
\end{enumerate}
\end{Cor}

According to Corollary~{\ref{cor:1}}, the presence or absence of hysteresis can be controlled by Spo0A$\sim$P or SigH when AbrB is functioning to a certain extent (Fig.~\ref{fig:3}B). 
Therefore, the question is whether the activity of Spo0A$\sim$P or SigH controls the environmental pH-dependent cell type hysteresis.
It is known that SigH is positively regulated with an increase in the environmental pH \citep{RN359,RN179}. 
Therefore, according to Corollary~{\ref{cor:1}}, if SigH is the input point for pH-dependent control, the hysteresis disappears at a low pH (gray dotted arrow in Fig.~\ref{fig:3}B). 
However, this finding contradicts the pH-dependent control phenomenon in actual colony observation (left part in Fig.~\ref{fig:3}A).

Note that AbrB has a stronger pH dependence than SigH \citep{RN179}. 
AbrB is upregulated as the environmental pH decreases at least in the pH 6-7 range. 
Therefore, we consider the possibility that the environmental-pH-dependent AbrB activity controls the cell type hysteresis.
To examine this aspect, we consider the mapping $\alpha_A \mapsto \hat{Q} - \hat{Q}_{\rm{c}}(\hat{P})$ $(\alpha_A \in (1,\infty ) )$. 
If $\alpha_A \rightarrow 1$, then $\hat{P} \rightarrow \infty$, and it follows from Theorem~{\ref{thm:1}} that 
\begin{equation*}
\hat{Q} - \hat{Q}_{\rm{c}}(\hat{P}) \sim \frac{1}{2} \left( \frac{\alpha_{SH}}{\alpha_A - 1} -1 \right) -  \frac{27}{2} \frac{\alpha_A^2}{\left( \alpha_A - 1\right)^4} \rightarrow - \infty . 
\end{equation*}
Similarly, we see that $\hat{P} \rightarrow +0$ as $\alpha_A \rightarrow \infty$. Theorem~{\ref{thm:1}} indicates that 
\begin{equation*}
\hat{Q} - \hat{Q}_{\rm{c}}(\hat{P}) \sim \frac{1}{2} \left( \frac{\alpha_{SH}}{\alpha_A - 1} -1 \right) -  \frac{3}{2} \frac{\alpha_A^{2/3}}{\left( \alpha_A - 1\right)^{4/3}}  \rightarrow - \frac{1}{2} . 
\end{equation*}
Furthermore, according to this theorem
\begin{equation*}
\frac{\partial}{\partial \alpha_A} \hat{Q}_{\rm{c}}'(\hat{P}) 
= -\hat{Q}_{\rm{c}}''(\hat{P}) \frac{\alpha_A + 1}{\left( \alpha_A - 1 \right)^3} <0 . 
\end{equation*}
Therefore, when $\alpha_A  \in (1,\infty )$ increases, $\hat{Q}_{\rm{c}}'(\hat{P})$ is positive and decreases monotonically. 
Consequently,
\begin{equation*}
\frac{\partial}{\partial \alpha_A} \left( \hat{Q} - \hat{Q}_{\rm{c}}(\hat{P}) \right) 
= \frac{1}{\left( \alpha_A - 1 \right)^2} \left\{ 
\left( 1+\frac{2}{\alpha_A - 1} \right) \hat{Q}_{\rm{c}}'(\hat{P}) - \frac{\alpha_{SH}}{2} 
\right\} 
\end{equation*}
changes its sign only once from positive to negative.  
In other words, there exists an $\alpha_{A,{\rm M}} \in (1,\infty )$ such that $\alpha_A \mapsto \hat{Q} - \hat{Q}_{\rm{c}}(\hat{P})$ increases when $\alpha_A < \alpha_{A,{\rm M}}$ and decreases when $\alpha_A > \alpha_{A,{\rm M}}$ (Fig.~\ref{fig:3}C). 

Accordingly, Theorem~{\ref{thm:1}} can be expressed as follows. 
\begin{Cor}
\label{cor:2}
There exists an $\alpha_{SH}^{\ast} > 0$ such that
\begin{enumerate}
\item if $\alpha_{SH} \leq \alpha_{SH}^{\ast}$, the set of steady states is anhysteretic; 
\item if $\alpha_{SH} > \alpha_{SH}^{\ast}$, there exist $1 < \alpha_{A,{\rm L}} < \alpha_{A,{\rm U}}$ such that 
\begin{enumerate}
\renewcommand{\labelenumii}{(\roman{enumii})}
\item if $0 < \alpha_A \leq \alpha_{A,{\rm L}}$, the set of steady states is anhysteretic. 
\item if $\alpha_{A,{\rm L}} < \alpha_A < \alpha_{A,{\rm U}}$, the set of steady states is hysteretic. 
\item if $\alpha_{A,{\rm U}} \leq \alpha_A$, the set of steady states is anhysteretic. 
\end{enumerate}
\end{enumerate}
\end{Cor}
This indicates that when $\alpha_{SH}$ is sufficiently large, that is, when Spo0A$\sim$P and SigH function to a reasonable extent, $\alpha_A$ or AbrB can control the hysteresis of the cell type selection (Fig.~\ref{fig:3}D).
Moreover, the cell type control is (i) anhysteretic (matrix producers dominant) when $\alpha_A$ is small; (ii) hysteretic (periodic) when $\alpha_A$ is intermediate; (iii) anhysteretic (motile cells dominant) when $\alpha_A$ is large. 
In terms of the effect of the environmental pH, $\alpha_A$ exhibits a negative correlation in the neutral range, that is, AbrB is upregulated at a low pH \citep{RN179}. 
In this case, the change in the cell type control from anhysteretic (matrix producer) to hysteric (periodic) with decreasing pH (Fig.~\ref{fig:3}A) can be explained by the upregulation of $\alpha_A$ (Fig.~\ref{fig:3}D). 
Therefore, it is suggested that AbrB plays a central role in cell type regulation in response to environmental pH changes. 
In colony formation, it is expected that the activity of AbrB is the key to the selection of the concentric pattern.

\section{Hysteretic and anhysteretic curves: Proof of Theorem~{\ref{thm:1}}}
\label{sec:proof}
Herein, we present the proof of Theorem~{\ref{thm:1}}. 
The equation for the steady state (equilibrium point) $(S, H, A, C)$ of \eqref{eq:SHAC} is as follows:
\begin{equation}
\label{eq:sSHAC}
    \left\{
    \begin{alignedat}{2}
	0 &= \frac{c_S H}{a_S + b_S C} - d_S S \\
	0 &= \frac{c_H}{a_H + b_H A} - d_H H \\
        0 &= \frac{c_A}{a_A + b_A S} - d_A A  \\
        0 &= c_C X A - d_C C
    \end{alignedat}
    \right.
\end{equation}
This system of equations contain a parameter $X$. 
We note that the set of all steady states together with the parameter $X$ is described by a single variable $\sigma = a_A + b_A S$: 
$(S, H, A, C, X)=(S(\sigma), H(\sigma), A(\sigma), C(\sigma), X(\sigma))$. 
Furthermore, since $\sigma = a_A + b_A S$, it can also be said that the other variables $(H, A, C, X)$ can be expressed in terms of $S$. 
If we eliminate $H$, $A$ and $C$ using these expressions, $X$ can be written as 
\begin{equation}
\label{eq:Xsigma}
\begin{aligned}
X &= \frac{\lambda \sigma^2}{\left( \sigma -p \right) \left( \sigma + q \right) } - \mu \sigma \\
&= -\frac{\sigma}{\left( \sigma -p \right) \left( \sigma + q \right) } \left[ \mu \sigma^2 - \left\{ \lambda + \mu \left( p-q \right) \right\} \sigma - \mu pq \right] 
\end{aligned}
\end{equation}
where $\sigma = a_A + b_A S$, 
\begin{equation}
\label{eq:pqlm}
\begin{aligned}
p &= a_A, &  q &= \frac{b_H c_A}{a_H d_A}, \\
\lambda &= \frac{c_S c_H b_A d_A d_C}{b_S d_S a_H d_H c_A c_C}, &\qquad 
\mu &= \frac{a_S d_A d_C}{b_S c_A c_C} . 
\end{aligned}
\end{equation}
Furthermore, by setting 
\begin{equation}
\mu \sigma^2 - \left\{ \lambda + \mu \left( p-q \right) \right\} \sigma - \mu pq = \mu \left( \sigma + a \right) \left( \sigma -b \right) , 
\end{equation}
$X$ can be expressed as 
\begin{equation}
\label{eq:Xsigma2}
X = -\frac{\mu \sigma \left( \sigma +a \right) \left( \sigma -b \right) }{\left( \sigma -p \right) \left( \sigma + q \right) } ,
\end{equation}
where, 
\begin{equation}
\label{eq:ab}
\begin{aligned}
\kappa &= \frac{\lambda}{\mu} = \frac{c_S c_H b_A}{a_S d_S a_H d_H}, \\
-a &= \frac{1}{2} \left\{ \kappa + p-q - \sqrt{\left( \kappa + p-q \right)^2 +4pq} \right\} <0 , \\
b &= \frac{1}{2} \left\{ \kappa + p-q + \sqrt{\left( \kappa + p-q \right)^2 +4pq} \right\} >0 . 
\end{aligned}
\end{equation}
In addition, the range of $\sigma$ for which $X>0$ is $p<\sigma <b$. 
Differentiating \eqref{eq:Xsigma2} with respect to $\sigma$ yields 
\begin{equation}
\label{eq:dXdsigma}
\frac{dX}{d\sigma} = \frac{\mu}{\left( \sigma -p \right)^2 \left( \sigma + q \right)^2} \varphi (\sigma ) , 
\end{equation}
where
\begin{equation}
\label{eq:varphi}
\begin{aligned}
\varphi (\sigma ) &= 
\sigma \left( \sigma +a \right) \left( \sigma -b \right) \left( 2\sigma +q-p \right) \\
&\quad -\left( \sigma -p \right) \left( \sigma +q \right) \left\{ \left( \sigma +a \right) \left( \sigma -b \right) + \sigma \left( \sigma -b \right) + \sigma \left( \sigma + a \right) \right\} . 
\end{aligned}
\end{equation}
This indicates that
\begin{equation}
\label{eq:varphipb}
\begin{aligned}
\varphi (p) &= p(p+a) (p-b) (p+q) < 0, \\
\varphi (b) &= -b (b-p) (b+q) (b+a) <0. 
\end{aligned}
\end{equation}
According to this expression, and because $\varphi '(p)=-2p^2\kappa < 0$, the quartic function $\varphi = \varphi (\sigma )$ can be classified into the following two types: 
\begin{enumerate}
\item $\varphi (\sigma ) \leq 0$ for $p < \sigma < b$. 
\item There exist $\sigma_{-}$ and $\sigma_{+}$ such that $p< \sigma_{-} < \sigma_{+} < b$, \\
$\varphi (\sigma )>0$ if $\sigma_{-} < \sigma < \sigma_{+}$, \\
$\varphi (\sigma ) \leq 0$ otherwise (i.e., $\sigma \leq \sigma_{-}$ or $\sigma \geq \sigma_{+}$). 
\end{enumerate}
Each case corresponds to an anhysteretic (Fig.~\ref{fig:1}B) and hysteretic case (Fig.~\ref{fig:1}C). 
Expanding \eqref{eq:varphi} yields 
\begin{equation}
\label{eq:varphi2}
\varphi (\sigma ) = 
-\sigma^4 + 2(p-q) \sigma^3 + \left\{ 3pq + (p-q) (a-b) -ab \right\} \sigma^2 +2(a-b) pq\sigma -abpq . 
\end{equation}
Because $b-a=p-q+\kappa$ and $ab=pq$,  
\begin{equation}
\label{eq:varphi3}
\varphi (\sigma ) = 
-\sigma^4 + 2(p-q) \sigma^3 + \left\{ 2pq - (p-q) (p-q+\kappa ) \right\} \sigma^2 -2(p-q+\kappa ) pq\sigma -p^2 q^2 . 
\end{equation}
Furthermore, because $D=p-q$ and $P=pq$, 
\begin{equation}
\label{eq:varphi4}
\varphi (\sigma ) = 
-\sigma^4 + 2D \sigma^3 + \left\{ 2P - D (D+\kappa ) \right\} \sigma^2 -2(D+\kappa ) P\sigma -P^2 . 
\end{equation}
As $p<\sigma <b$, the range of the variable $\sigma$ is 
\begin{equation}
\label{eq:domainofvarphi}
\frac{1}{2} \left( D+\sqrt{D^2+4P} \right) < \sigma < \frac{1}{2} \left( D+\kappa + \sqrt{\left( D+\kappa \right)^2 +4P} \right) . 
\end{equation}
This treatment indicates that the hysteretic or anhysteretic nature of the set of steady states depends on whether the function $\varphi (\sigma )$, defined by \eqref{eq:varphi4} in the range \eqref{eq:domainofvarphi}, does or does not (non-positive) undergo a sign change, respectively.
In the following text, we examine the cases of $D \geq 0$ and $D<0$ separately.

\subsection{No hysteresis under low influence of AbrB}
\label{subsec:AbrBlow}
First, we demonstrate that hysteresis does not occur when AbrB has a low influence, that is, $D\geq 0$ ($\alpha_{A} \leq 1$).
When $D=0$ ($\alpha_{A}=1$), \eqref{eq:domainofvarphi} is reduced to 
\begin{equation}
\label{eq:domainofvarphiDzero}
\sqrt{P} < \sigma < \frac{1}{2} \left( \kappa + \sqrt{\kappa^2 +4P} \right) , 
\end{equation}
and 
\begin{equation}
\label{eq:varphiDzero}
\varphi (\sigma ) = -\left( \sigma^2 - P \right)^2 -2\kappa P \sigma <0 . 
\end{equation}
When $D>0$ ($\alpha_{A}<1$), 
if we scale $P$, $\kappa$, $\sigma$ by $D=-\eta$ so that $P=\eta^2 \hat{P}$, $\kappa = \eta\hat{\kappa}$, $\sigma = \eta\hat{\sigma}$, 
then $\eta <0$, $\hat{P}>0$, $\hat{\kappa}<0$, $\hat{\sigma}<0$, and \eqref{eq:domainofvarphi} becomes 
\begin{equation}
\frac{1}{2} \left( -\eta + \sqrt{\eta^2 + 4\eta^2 \hat{P}} \right) < \eta \hat{\sigma} < \frac{1}{2} \left\{ -\eta + \eta\hat{\kappa} + \sqrt{\left( -\eta + \eta\hat{\kappa} \right)^2 + 4\eta^2\hat{P}} \right\} . 
\end{equation}
Dividing this expression by $\eta (<0)$ yields 
\begin{equation}
\frac{1}{2} \left\{ -1+\hat{\kappa} - \sqrt{\left( 1-\hat{\kappa} \right)^2 +4\hat{P}} \right\} < \hat{\sigma} < \frac{1}{2} \left( -1-\sqrt{1+4\hat{P}
} \right) . 
\end{equation}
Setting $\hat{Q} = \frac{1}{2}(\hat{\kappa} - 1)$, we see that $\hat{\kappa} < 0$ implies $\hat{Q} < -1/2$, and the domain can be expressed as 
\begin{equation}
\hat{Q} -\sqrt{\hat{Q}^2+4\hat{P}} < \hat{\sigma} < \frac{1}{2} \left( -1-\sqrt{1+4\hat{P}} \right). 
\end{equation}

Furthermore, by considering 
\begin{equation}
\label{eq:hofzDp}
h(z) = 
z^4 - 2z^3 - 2\left( \hat{P}+\hat{Q} \right) z^2 -4\hat{P}\hat{Q}z +\hat{P}^2  , 
\end{equation}
we can obtain $\varphi (\sigma ) = -\eta^4 h(-\hat{\sigma})$. 
As $p<\sigma <b$, the range of the variable $z = -\hat{\sigma}$ is 
\begin{equation}
\label{eq:domainofhDp}
\frac{1}{2} \left(  1 + \sqrt{1+4\hat{P}} \right) < z < -\hat{Q}+\sqrt{\hat{Q}^2 + 4\hat{P}}. 
\end{equation}
Therefore, it is sufficient to show that the function $h(z)$, defined by \eqref{eq:hofzDp} in \eqref{eq:domainofhDp}, is positive. 
From $-2\hat{Q}>1$, it follows that 
\begin{equation}
\begin{aligned}
h(z) &= z^2 \left\{ z^2 -2z+\left( -2\hat{Q} \right) \right\} + 2\hat{P}z \left\{ -z+\left( -2\hat{Q} \right) \right\} +\hat{P}^2 \\
 &> z^2 \left( z^2 -2z+1 \right) + 2\hat{P}z \left(-z+1 \right) + \hat{P}^2 \\
 &= z^2 (z-1)^2 -2\hat{P}z(z-1) +\hat{P}^2 \\
 &= \left( z^2 -z-\hat{P} \right)^2  > 0 .
\end{aligned}
\end{equation}
Therefore, no hysteresis occurs.

\subsection{Hysteresis condition under high influence of AbrB}
\label{subsec:AbrBhigh}
Herein, we consider in detail the conditions under which the steady state becomes hysteretic when AbrB has a strong influence, that is, when $D<0$ ($\alpha_{A}>1$). 
Substituting $D=-\eta$, $P=\eta^2 \hat{P}$, $\kappa = \eta\hat{\kappa}$, $\sigma = \eta\hat{\sigma}$ yields 
\begin{equation}
\label{eq:varphi5}
\varphi (\sigma ) = \eta^4 \left\{ 
-\hat{\sigma}^4 - 2 \hat{\sigma}^3 + 2\left( \hat{P}+\hat{Q} \right) \hat{\sigma}^2 -4\hat{P}\hat{Q}\hat{\sigma} -\hat{P}^2  
\right\} , 
\end{equation}
where $\hat{Q}=\frac{1}{2}(\hat{\kappa}-1) > -1/2$ since $\hat{\kappa} > 0$. Hence, considering 
\begin{equation}
\label{eq:hofz}
h(z) = 
z^4 + 2z^3 - 2\left( \hat{P}+\hat{Q} \right) z^2 +4\hat{P}\hat{Q}z +\hat{P}^2  
\end{equation}
yields $\varphi (\sigma ) = -\eta^4 h(\hat{\sigma})$. 
Because $p<\sigma <b$, the range of the variable $z = \hat{\sigma}$ is
\begin{equation}
\label{eq:domainofh}
\frac{1}{2} \left( \sqrt{1+4\hat{P}}-1 \right) < z < \sqrt{\hat{Q}^2 + \hat{P}} + \hat{Q} . 
\end{equation}
Thus, the hysteretic or anhysteretic nature of the set of steady states depends on whether the function $h(z)$, defined as in \eqref{eq:hofz} in the domain \eqref{eq:domainofh}, does or does not (non-negative) undergo a sign change. 

Because
\begin{equation}
\label{eq:hofz2}
h(z) = 
2\hat{Q}z \left( 2\hat{P}-z \right) + \left( z^2 - \hat{P} \right)^2 + 2z^3, 
\end{equation}
$h(z) >0$ for $z \leq 2\hat{P}$. 
In particular, $h$ has a fixed value independent of $\hat{Q}$ at $z = 2\hat{P}$: $h(2\hat{P}) = \hat{P}^2 ( 4\hat{P}+1 )^2$. 
Conversely, in view of 
\begin{equation}
h(z)<0 \quad \Leftrightarrow \quad 
z^4 + 2z^3 - 2\hat{P}z^2 +\hat{P}^2 < 2z \left( z-2\hat{P} \right) \hat{Q}, 
\end{equation}
for each $z>2\hat{P}$, we see that 
\begin{equation}
\label{eq:qtilde}
h(z)<0 \quad \Leftrightarrow \quad 
\hat{Q} > \frac{\left( z^2 -\hat{P} \right)^2 + 2z^3}{2z \left( z-2\hat{P} \right)} =: \widetilde{Q}\left( z; \hat{P} \right) . 
\end{equation}
Hence, $h(z)$ is negative if $\hat{Q}$ is sufficiently large, and the set of steady states is hysteretic. 
In contrast, if $\hat{Q}\leq 0$, $h(z)$ is always non-negative, and the set of steady states is anhysteretic. 
Furthermore, 
\begin{equation}
\frac{\partial h}{\partial \hat{Q}} = -2z \left( z - 2\hat{P} \right)
\end{equation}
is negative if $z>2\hat{P}$. 
Therefore, for each $\hat{P}$, there exists a threshold $\hat{Q}_{\rm{c}}(\hat{P}) > 0$ such that
\begin{enumerate}
\item if $\hat{Q} \leq \hat{Q}_{\rm{c}}(\hat{P})$, the set of steady states is anhysteretic;  
\item if $\hat{Q} > \hat{Q}_{\rm{c}}(\hat{P})$, the set of steady states is hysteretic.  
\end{enumerate}
This threshold $\hat{Q}_{\rm{c}}(\hat{P})$ can be understood as the minimum value of the function  $\widetilde{Q}\left(\cdot; \hat{P} \right) $ given in \eqref{eq:qtilde} as follows. 
For any $\hat{P}>0$, the function $\widetilde{Q}\left( \cdot; \hat{P} \right)=\widetilde{Q}\left( z; \hat{P} \right)$ on $z > 2\hat{P}$ achieves a unique minimum value $\hat{Q}_{\rm{c}}(\hat{P})$ at $z=z_{\rm{c}}(\hat{P})$, which is a unique critical point (Fig.~\ref{fig:4}).
This minimum value $\hat{Q}_{\rm{c}}(\hat{P})$ is precisely the threshold of $\hat{Q}$ for hysteresis. 

\begin{figure*}[tbp]
\begin{center}
\includegraphics[width=72mm]{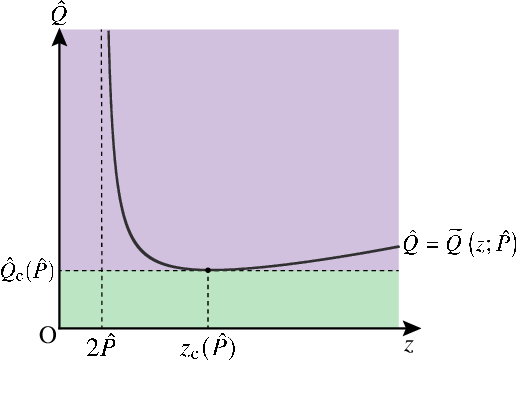}
\end{center}
\caption{Curve $\widetilde{Q}\left( \cdot; \hat{P} \right)$ and the hysteresis threshold $\hat{Q}_{\rm{c}}(\hat{P})$. }. 
\label{fig:4}
\end{figure*}

\subsection{Profile of the hysteresis threshold curve}
\label{subsec:profilehtc}
We present the profile \eqref{eq:qcprofile} of the threshold function $\hat{Q}_{\rm{c}} = \hat{Q}_{\rm{c}}(\hat{P})$ to complete the proof of Theorem~{\ref{thm:1}}. 
First, the threshold curve is characterized as $\left( z, \hat{Q} \right) = \left( z_{\rm{c}}(\hat{P}), \hat{Q}_{\rm{c}} (\hat{P}) \right)$, which satisfies the following two equations for $\hat{P}$:
\begin{align}
\label{Fzero}
F \left( z, \hat{P}, \hat{Q} \right) &= z^4 + 2z^3 - 2\left( \hat{P}+\hat{Q} \right) z^2 +4\hat{P}\hat{Q} z +\hat{P}^2 = 0, \\ 
\label{Fzzero}
F_z \left( z, \hat{P}, \hat{Q} \right) &= 4 z^3 + 6 z^2 -4 \left( \hat{P} + \hat{Q} \right) z + 4 \hat{P} \hat{Q} = 0. 
\end{align}
Then, \eqref{eq:qcprofile} follows from the following three lemmas.

\begin{Lem}
\label{lem:Qcprime}
\begin{equation}
\label{eq:Qcprime}
 \hat{Q}_{\rm{c}}' (\hat{P}) = - \frac{F_{\hat{P}}}{ F_{\hat{Q}} } > 0, \quad 
\hat{Q}_{\rm{c}}'' (\hat{P}) = - \frac{ F_{ \hat{P}\hat{P} } F_{\hat{Q}} - F_{\hat{P}} F_{\hat{Q} \hat{P}} }{F_{\hat{Q}}^2} > 0. 
\end{equation}
Here, $ \left( z_{\rm{c}}(\hat{P}), \hat{P }, \hat{Q}_{\rm{c}} (\hat{P}) \right)$ is omitted on the right side, for example, $F_{\hat{P}} = F_{\hat{P}} \left( z_{\rm{c}}(\hat{P}), \hat{P}, \hat {Q}_{\rm{c}} (\hat{P}) \right)$. 
\end{Lem}

\begin{Lem}
\label{lem:asymptPzero}
As $\hat{P} \rightarrow 0$, the following holds. 
\begin{equation}
\label{eq:asmptPzero}
\frac{z_{\rm{c}}(\hat{P})}{\hat{P}^{2/3}} \rightarrow 1, \quad 
\frac{\hat{Q}_{\rm{c}}(\hat{P})}{\hat{P}^{2/3}} \rightarrow \frac{3}{2}. 
\end{equation}
\end{Lem}

\begin{Lem}
\label{lem:asymptPinfty}
As $\hat{P} \rightarrow \infty$, the following holds. 
\begin{equation}
\label{eq:asmptPinfty}
\frac{z_{\rm{c}}(\hat{P})}{\hat{P}} \rightarrow 3, \quad 
\frac{\hat{Q}_{\rm{c}}(\hat{P})}{\hat{P}^{2}} \rightarrow \frac{27}{2}. 
\end{equation}
\end{Lem}

\begin{proof}[Proof of Lemma~{\ref{lem:Qcprime}}]
By differentiating $F \left( z_{\rm{c}}(\hat{P}), \hat{P}, \hat{Q}_{\rm{c}} (\hat{P}) \right) =0$ with respect to $\hat{P}$, \eqref{eq:Qcprime} can be obtained. 
In particular, we can show that 
\begin{alignat}{2}
F_{\hat{P}} &= -2z^2 + 4 \hat{Q} z + 2 \hat{P} > 0, &\quad F_{\hat{P}\hat{P}} &= 2 > 0, \\ 
F_{\hat{Q}} &= -2z \left( z - 2\hat{P} \right) < 0, &\quad F_{\hat{Q}\hat{P}} &= 4z > 0 
\end{alignat}
for $\left( z, \hat{Q} \right) = \left( z_{\rm{c}}(\hat{P}), \hat{Q}_{\rm{c}} (\hat{P}) \right)$. 
$z - 2\hat{P} >0$ was examined in the previous subsection~\ref{subsec:AbrBhigh}. 
Furthermore, we can obtain $z^2 - 2 \hat{Q} z - \hat{P} < 0$, as follows. 
Assume that $z  \geq \hat{Q}$. According to \eqref{eq:domainofh}, $z - \hat{Q} < \sqrt{\hat{Q}^2 + \hat{P}}$. 
If both sides are squared, we obtain 
$\left( z - \hat{Q} \right)^2 < \hat{Q}^2 + \hat{P}$ as $z - \hat{Q} \geq 0$. 
Rearranging this expression yields $z^2 - 2 \hat{Q} z - \hat{P} < 0$. 
However, when $z  < \hat{Q}$, we have $\hat{Q}>0$ and 
$z^2 - 2 \hat{Q} z - \hat{P} = -z \left( \hat{Q} - z \right) - \hat{Q} z - \hat{P} < 0$. 
\end{proof}

\begin{proof}[Proof of Lemma~{\ref{lem:asymptPzero}}]
Substituting $z=\alpha \hat{P}^{2/3}$ into \eqref{Fzero} and \eqref{Fzzero} yields 
\begin{align}
&\hat{P}^{4/3} \left( \alpha^4 \hat{P}^{4/3} + 2\alpha^3 \hat{P}^{2/3} -2\alpha^2 \hat{P}  -2\alpha^2 \hat{Q} +4\alpha \hat{P}^{1/3} \hat{Q} + \hat{P}^{2/3} \right) = 0, \\ 
&2\hat{P}^{2/3} \left( 2\alpha^3 \hat{P}^{4/3} +3\alpha^2\hat{P}^{2/3} -2\alpha\hat{P} -2\alpha\hat{Q} + 2\hat{P}^{1/3}\hat{Q} \right) = 0, 
\end{align}
respectively. Because $\hat{P}>0$, $\alpha >0$, 
\begin{align}
\hat{Q} &= \frac{\alpha^4 \hat{P}^{4/3} + 2\alpha^3 \hat{P}^{2/3} -2\alpha^2 \hat{P}  + \hat{P}^{2/3}}{2\alpha \left( \alpha - 2 \hat{P}^{1/3} \right)}, \\ 
\hat{Q} &= \frac{2\alpha^3 \hat{P}^{4/3} +3\alpha^2\hat{P}^{2/3} -2\alpha\hat{P}}{2 \left( \alpha - \hat{P}^{1/3} \right)} . 
\end{align}
In this case, we assumed that $\alpha > 2\hat{P}^{1/3}$, which is true for small $\hat{P}$ if $\alpha \rightarrow \alpha_0 >0$ as $\hat{P} \rightarrow 0$. 
Therefore, 
\begin{equation}
\alpha = 
\frac{\alpha^4 \hat{P}^{4/3} + 2\alpha^3 \hat{P}^{2/3} -2\alpha^2 \hat{P}  + \hat{P}^{2/3}}{2\alpha^3 \hat{P}^{4/3} +3\alpha^2\hat{P}^{2/3} -2\alpha\hat{P}} \cdot \frac{\alpha - \hat{P}^{1/3}}{\alpha - 2\hat{P}^{1/3}} . 
\end{equation}
If we assume that $\alpha \rightarrow \alpha_0$ as $\hat{P} \rightarrow 0$, then 
\begin{equation}
\alpha_0 = \frac{2\alpha_0^3 + 1}{3\alpha_0^2} . 
\end{equation}
Hence $\alpha_0 = 1$, and \eqref{eq:asmptPzero} can be obtained. 

The above arguments are based on the assumption that there exists a limit of $\alpha (\hat{P} )$. 
The implicit function theorem can be used to justify these arguments. 
\end{proof}

\begin{proof}[Proof of Lemma~{\ref{lem:asymptPinfty}}]
Substituting $z=\alpha \hat{P}$ into \eqref{Fzero} and \eqref{Fzzero} yields  
\begin{align}
&\hat{P}^{2} \left( \alpha^4 \hat{P}^{2} + 2\alpha^3 \hat{P} -2\alpha^2 \hat{P}  -2\alpha^2 \hat{Q} +4\alpha \hat{Q} + 1 \right) = 0, \\ 
&2\alpha^3 \hat{P}^3 +3\alpha^2\hat{P}^2 -2\alpha\hat{P}^2 -2( \alpha -1) \hat{P}\hat{Q}  = 0, 
\end{align}
respectively. Because $\hat{P}>0$, $\alpha >0$, 
\begin{align}
&\alpha^4 \hat{P}^{2} + 2\alpha^3 \hat{P} -2\alpha^2 \hat{P}  + 1 = 2\alpha (\alpha -2) \hat{Q}, \\ 
&2\alpha^3 \hat{P}^2 + \alpha (3\alpha -2) \hat{P}  = 2( \alpha -1) \hat{Q}. 
\end{align}
If we assume that $\alpha > 2$, 
\begin{align}
\hat{Q} &= \frac{\alpha^2 \left\{ \alpha^2 \hat{P}^2 + 2 (\alpha -1) \hat{P} \right\} +1}{2\alpha \left( \alpha - 2 \right)}, \\ 
\hat{Q} &= \frac{\alpha \left\{ 2\alpha^2 \hat{P}^2 + (3\alpha -2) \hat{P}\right\}}{2 \left( \alpha -1 \right)} . 
\end{align}
Therefore, 
\begin{equation}
\frac{\alpha (\alpha -2)}{\alpha -1} = 
\frac{\alpha^2 \left\{ \alpha^2 \hat{P}^2 + 2\left( \alpha - 1 \right) \hat{P} \right\} +1}{\alpha \left\{ 2\alpha^2\hat{P}^2 + \left( 3\alpha -2 \right) \hat{P}\right\} } . 
\end{equation}
If we assume that $\alpha \rightarrow \alpha_\infty$ as $\hat{P} \rightarrow \infty$, then 
\begin{equation}
\frac{\alpha_\infty ( \alpha_\infty -2)}{\alpha_\infty -1} = \frac{\alpha_\infty}{2}. 
\end{equation}
Since we have assumed that $\alpha >2$, necessarily we have $\alpha_\infty \geq 2$. In this case, $\alpha_\infty = 3$, and \eqref{eq:asmptPinfty} can be obtained. 

Similar to Lemma~{\ref{lem:asymptPzero}}, the implicit function theorem can be used to justify these arguments. 
\end{proof}

\section{Stability analysis}
\label{sec:stability}
We present the proofs pertaining to the stability and instability of the equilibrium points. 
The stability of an equilibrium point $(S, H, A, C)$ in \eqref{eq:SHAC} indicates that the real parts of all the roots of the following characteristic equation are negative: 
\begin{equation}
\det 
\begin{pmatrix}
  -d_S-\lambda & p_S & 0 & -n_S \\
  0 & -d_H-\lambda & -n_H & 0 \\
  -n_A & 0 & -d_A-\lambda & 0 \\
  0 & 0 & p_C & -d_C-\lambda 
\end{pmatrix}
= 0 ,  
\end{equation}
where 
\begin{alignat}{3}
n_S &= \frac{b_S c_S H}{\left( a_S + b_S C \right)^2}, &\quad
n_H &= \frac{b_H c_H}{\left( a_H + b_H A \right)^2}, \quad 
n_A = \frac{b_A c_A}{\left( a_A + b_A S \right)^2}, \\
p_S &= \frac{c_S}{a_S + b_S C}, &\quad 
p_C &= c_C X
\end{alignat}
are all positive constants (except $p_C =0$ when $X=0$). 
If we expand the characteristic equation as 
\begin{equation}
\label{eq:ce}
\lambda^4 + a_3 \lambda^3 + a_2 \lambda^2 + a_1 \lambda + a_0 = 0 , 
\end{equation}
the coefficients are as follows: 
\begin{equation}
\label{eq:ai}
\begin{aligned}
a_0 &= d_S d_H d_A d_C - \left( p_S n_A n_H d_C + p_C n_A n_S d_H \right), \\
a_1 &= d_S d_H \left( d_A + d_C \right) + d_A d_C \left( d_S + d_H \right) - \left( p_S n_A n_H + p_C n_A n_S \right), \\
a_2 &= d_S d_H + d_A d_C + \left( d_S + d_H \right) \left( d_A + d_C \right) ,  \\
a_3 &= d_S + d_H + d_A + d_C . 
\end{aligned}
\end{equation}
Then, according to the Routh--Hurwitz criterion, the necessary and sufficient condition for the real parts of $\lambda_1$, $\lambda_2$, $\lambda_3$, and $\lambda_4$ to be negative is for the following six inequalities to be satisfied: 
$a_0>0$, $a_1>0$, $a_2>0$, $a_3 > 0$, 
\begin{equation*}
\det 
\begin{pmatrix}
 a_3 & a_1 \\
 1 & a_2 
\end{pmatrix}
= a_2 a_3 - a_1 >0, \quad 
\det 
\begin{pmatrix}
 a_3 & a_1 & 0 \\
 1 & a_2 & a_0 \\
 0 & a_3 & a_1
\end{pmatrix}
= a_1 a_2 a_3 - a_0 a_3^2 - a_1^2 > 0 . 
\end{equation*}
Among these relations, $a_2 > 0$, $a_3 > 0$ and $a_2 a_3 - a_1 >0$ always hold. Therefore, the stability is determined by the remaining three conditions. 
Furthermore, the following lemma holds. 
\begin{Lem}
\label{lem:B1}
If $a_0 \geq 0$, then $a_1 a_2 a_3 - a_0 a_3^2 - a_1^2 > 0$. 
\end{Lem}
\noindent 
Overall, only two conditions remain for stability: $a_0 >0$, $a_1 >0$. 
On this basis, we can state that 
\begin{enumerate}
\item Any equilibrium point is stable when the equilibrium curve $X=X(S)$ is anhysteretic (Fig.~{\ref{fig:1}}B).
\item When the curve of the equilibrium point is hysteretic, the equilibrium point between the two folding points (point at $X'(S)=0$) is unstable, and the outer equilibrium point is stable (Fig.~{\ref{fig:1}}C).
\end{enumerate}
To prove this aspect, it is sufficient to demonstrate the following four lemmas.

\begin{Lem}
\label{lem:B2}
$\operatorname{sgn} \lambda_1 \lambda_2 \lambda_3 \lambda_4 = - \operatorname{sgn} \varphi (\sigma )$. 
\end{Lem}
\noindent 
Here, $\varphi (\sigma )$ $(\sigma = a_A + b_A S)$ is the function defined by \eqref{eq:varphi} representing the increase/decrease in the curve $X=X(S)$ of the equilibrium points, and $\operatorname{sgn} X'(S) = \operatorname{sgn} \varphi(\sigma )$. 
In other words, according to Lemma~{\ref{lem:B2}}, the zero eigenvalue appears at the turning point of the curve. 
Moreover, the sign of $a_0 = \lambda_1 \lambda_2 \lambda_3 \lambda_4$ changes through the turning point. Specifically, if the system is stable from a certain point to the turning point, it becomes unstable at the turning point.

\begin{Lem}
\label{lem:B3}
If the real parts of all the eigenvalues $\lambda_1$, $\lambda_2$, $\lambda_3$, and $\lambda_4$ are non-positive, none of them are pure imaginary numbers. 
\end{Lem}
This Lemma indicates that when the stability changes along the curve of the equilibrium points, it must pass through the zero eigenvalue and not the pure imaginary number. 
Moreover, according to Lemma~{\ref{lem:B2}}, the point at which the stability changes is the turning point of the curve of the equilibrium points.

\begin{Lem}
\label{lem:B4}
If $X$ is sufficiently large, any equilibrium point is stable. 
\end{Lem}

\begin{Lem}
\label{lem:B5}
If $X=0$, any equilibrium point is stable. 
\end{Lem}
\noindent 
Thus, the equilibrium point is always stable if no folding points occur. 
When there are two turning points $(X_1 , S_1 )$ and $(X_2 , S_2 )$ with $X_1 < X_2$, all equilibrium points satisfying $S_1 < S < S_2$ are unstable, and the other equilibrium points are stable.

\begin{proof}[Proof of Lemma~{\ref{lem:B1}}]
By substituting $\tilde{p}_S = p_S n_A n_H$ and $\tilde{p}_C = p_C n_A n_S$, we have 
\begin{align*}
 &\qquad a_1 a_2 a_3 - a_0 a_3^2 - a_1^2 \\
 &= \left\{ \left( d_C + d_A \right) d_H^2 + \left( d_C^2 + 2d_A d_C + d_A^2 \right) d_H + d_A d_C^2 + d_A^2 d_C \right\} d_S^3 \\
 &\quad + \bigl\{ -\left( d_H + d_A \right) \tilde{p}_S -\left( d_C + d_A \right) \tilde{p}_C + \left( d_C + d_A \right) d_H^3 + \left( 2d_C^2 + 4 d_A d_C + 2 d_A^2 \right) d_H^2 \\
 &\quad + \left( d_C^3 + 4 d_A d_C^2 + 4 d_A^2 d_C + d_A^3 \right) d_H + d_A d_C^3 + 2 d_A^2 d_C^2 + d_A^3 d_C \bigr\} d_S^2 \\
 &\quad + \Bigl[ \left\{ -d_H^2 + \left( d_C - d_A \right) d_H + d_C^2 + d_A d_C - d_A^2 \right\} \tilde{p}_S \\
 &\quad + \left\{ d_H^2 + \left( d_A + d_C \right) d_H - d_A^2 - d_A d_C - d_C^2 \right\} \tilde{p}_C \\
 &\quad + \left( d_C^2 + 2 d_A d_C + d_A^2 \right) d_H^3 + \left( d_C ^3 + 4 d_A d_C^2 + 4d_A^2 d_C + d_A^3 \right) d_H^2 \\
 &\quad + \left( 2 d_A^3 d_C + 4 d_A^2 d_C^2 +2d_A d_C^3 \right) d_H + d_C^3 d_A^2 + d_C^2 d_A^3 \Bigr] d_S \\
 &\quad -\tilde{p}_S^2 + \left\{ -2\tilde{p}_C - d_A d_H^2 + \left( d_C^2 + d_A d_C - d_A^2 \right) d_H + d_C^3 + d_A d_C^2 \right\} \tilde{p}_S \\
 &\quad -\tilde{p}_C^2 + \left\{ d_H^3 + \left( d_C + d_A \right) d_H^2 + d_A d_C d_H - d_A d_C^2 -d_A^2 d_C \right\} \tilde{p}_C \\
 &\quad + \left( d_A d_C^2 + d_A^2 d_C \right) d_H^3 + \left( d_A d_C^3 + 2 d_A^2 d_C^2 + d_A^3 d_C \right) d_H^2 + \left( d_A^2 d_C^3 + d_A^3 d_C^2 \right) d_H . 
\end{align*}
As $a_0 \geq 0$, $\tilde{p}_S \leq d_S d_H d_A$, $\tilde{p}_C \leq d_S d_A d_C$. Using these expressions for the negative terms in the abovementioned equation, we can obtain the following expression: 
\begin{align*}
 &\qquad a_1 a_2 a_3 - a_0 a_3^2 - a_1^2 \\
 &\geq \left\{ d_C d_H^2 + \left( d_C^2 + 2d_A d_C \right) d_H \right\} d_S^3 \\
 &\quad + \bigl\{ d_C d_H^3 + \left( 2d_C^2 + 4 d_A d_C \right) d_H^2 + \left( d_C^3 + 4 d_A d_C^2 + 4 d_A^2 d_C \right) d_H \bigr\} d_S^2 \\
 &\quad + \Bigl[ \left\{ d_C d_H + d_C^2 + d_A d_C \right\} \tilde{p}_S + \left\{ d_H^2 + \left( d_A + d_C \right) d_H \right\} \tilde{p}_C \\
 &\quad + \left( d_C^2 + 2 d_A d_C \right) d_H^3 + \left( d_C ^3 + 4 d_A d_C^2 + 4d_A^2 d_C \right) d_H^2 \\
 &\quad + \left( 2 d_A^3 d_C + 2 d_A^2 d_C^2 +2d_A d_C^3 \right) d_H \Bigr] d_S \\
 &\quad + \left\{ \left( d_C^2 + d_A d_C \right) d_H + d_C^3 + d_A d_C^2 \right\} \tilde{p}_S + \left\{ d_H^3 + \left( d_C + d_A \right) d_H^2 + d_A d_C d_H \right\} \tilde{p}_C \\
 &\quad + \left( d_A d_C^2 + d_A^2 d_C \right) d_H^3 + \left( d_A d_C^3 + 2 d_A^2 d_C^2 + d_A^3 d_C \right) d_H^2 + \left( d_A^2 d_C^3 + d_A^3 d_C^2 \right) d_H , 
\end{align*}
the right-hand side of which is positive.  
\end{proof}

\begin{proof}[Proof of Lemma~{\ref{lem:B2}}]
By using the equations of the equilibrium point, \eqref{eq:sSHAC}, in \eqref{eq:ai} to eliminate $H$, $A$, $C$ and $X$, and by rearranging the equation by using $\sigma$ $(= a_A + b_A S)$, $\alpha_A$, $\alpha_{SH}$, we can obtain 
\begin{equation}
\label{eq:a0s}
a_0 = -\frac{d_S d_H d_A d_C}{a_A \sigma \left( \sigma + a_A \alpha_A \right) \alpha_{SH}} \varphi (\sigma) . 
\end{equation}
Therefore, from $a_0 = \lambda_1 \lambda_2 \lambda_3 \lambda_4$, the claim is verified. 
\end{proof}

\begin{proof}[Proof of Lemma~{\ref{lem:B3}}]
The necessary and sufficient condition for the characteristic equations \eqref{eq:ce} to have a pure imaginary root $\lambda = \rho i$ $(\rho \neq 0)$ is $a_1 a_2 a_3 - a_0 a_3^2 - a_1^2 = 0$ because $\rho^4 - a_2 \rho^2 + a_0 =0$ and $\rho \left( a_3 \rho^2 -a_1 \right) =0$. 
According to the assumption, because the Routh--Hurwitz conditions hold with equal signs, we have $a_0 \geq 0$. 
Therefore, according to Lemma~{\ref{lem:B1}}, $a_1 a_2 a_3 - a_0 a_3^2 - a_1^2 > 0$. 
Thus, the characteristic equation \eqref{eq:ce} does not have any pure imaginary root. 
\end{proof}

\begin{proof}[Proof of Lemma~{\ref{lem:B4}}]
As mentioned previously, only two of the Routh--Hurwitz conditions must be examined: $a_0> 0$ and $a_1 >0$. 
When $X$ is sufficiently large, $S \sim 0$ $(\sigma \sim a_A)$. 
Then, the equilibrium point satisfies $X'(S)<0$, and therefore, $\varphi (\sigma ) <0$. 
Thus, $a_0 > 0$ immediately follows from \eqref{eq:a0s}. 
Next, considering $a_1$, in the same manner as \eqref{eq:a0s}, by rearranging the expression using $\sigma$, $\alpha_A$ and $\alpha_{SH}$, we can obtain 
\begin{equation}
\label{eq:a1s}
\begin{aligned}
a_1 &= d_S d_H \left( d_A + d_C \right) + d_A d_C \left( d_S + d_H \right) 
- d_S d_H d_A \frac{a_A \alpha_A \left( \sigma - a_A \right)}{\sigma \left( \sigma + a_A \alpha_A \right)} \\
&\quad -d_A d_C d_S \frac{\sigma - a_A}{\sigma} \left\{ 1- \frac{\left( \sigma  - a_A \right) \left( \sigma + a_A \alpha_A \right)}{\alpha_{SH}a_A \sigma} \right\} . 
\end{aligned}
\end{equation}
Because $\sigma \sim a_A$, we see that $a_1 \sim d_S d_H \left( d_A + d_C \right) + d_A d_C \left( d_S + d_H \right) >0$. 
\end{proof}

\begin{proof}[Proof of Lemma~{\ref{lem:B5}}]
As in the previous case, it is sufficient to consider only two conditions: $a_0 > 0$ and $a_1 >0$. 
When $X=0$, we have $\sigma = b$. 
Then, the equilibrium point satisfies $X'(S)<0$, and therefore, $\varphi (\sigma ) <0$. 
Thus, $a_0 > 0$ immediately follows from \eqref{eq:a0s}. 
By rearranging $\sigma = b$ using $\alpha_A$ and $\alpha_{SH}$, we obtain  
\begin{equation}
\label{eq:bstar}
\begin{aligned}
b &= \eta \left( \sqrt{\hat{Q}^2 + \hat{P}} + \hat{Q} \right) \\
 &= p \left( \frac{q}{p} - 1 \right) \left\{ \frac{1}{2} \sqrt{ \left( \frac{\alpha_{SH} - \alpha_A +1}{\alpha_A - 1} \right)^2 + \frac{\alpha_A}{\left( \alpha_A - 1 \right)^2} } 
 + \frac{1}{2} \left( \frac{\alpha_{SH} - \alpha_A +1}{\alpha_A - 1} \right) \right\} \\
 &= a_A \left( \alpha_A -1 \right) \left\{ \frac{1}{2} \sqrt{ \left( \frac{\alpha_{SH} - \alpha_A +1}{\alpha_A - 1} \right)^2 + \frac{\alpha_A}{\left( \alpha_A - 1 \right)^2} } 
 + \frac{1}{2} \left( \frac{\alpha_{SH} - \alpha_A +1}{\alpha_A - 1} \right) \right\} \\
 &=: a_A \left( \alpha_A -1 \right) b^{\ast} . 
\end{aligned}
\end{equation}
Here, we can see that $b^{\ast}$ satisfies the following condition 
\begin{equation}
\label{eq:bstarineq}
b^{\ast} \geq \frac{1}{\alpha_A - 1} , \quad 
b^{\ast} \geq \frac{\alpha_{SH} - \alpha_A +1}{\alpha_A - 1}
\end{equation}
as $b > p = a_A$. 
Then, from 
\begin{equation}
p_S = \frac{c_S}{a_S}, \quad 
n_H = \frac{b_H c_H}{\left( a_H + b_H \cfrac{c_A}{d_A} \cfrac{1}{\sigma}\right)^2}, \quad 
n_A = \frac{b_A c_A}{\sigma^2}, 
\end{equation}
we have 
\begin{equation}
\begin{aligned}
\frac{p_S n_H n_A}{d_A d_S d_H} &= \frac{\alpha_{SH} \alpha_A}{\left( \cfrac{\sigma}{a_A} + \alpha_A \right)^2} \\
&= \frac{\alpha_{SH} \alpha_A}{\left\{ b^{\ast} \left( \alpha_A -1 \right)+ \alpha_A \right\}^2} . 
\end{aligned}
\end{equation}
By using the two inequalities \eqref{eq:bstarineq} at the square of the denominator on the right-hand side, one can obtain
\begin{equation}
\frac{p_S n_H n_A}{d_A d_S d_H} \leq \frac{\alpha_{SH}}{1+ \alpha_{SH}} \frac{\alpha_A}{1+ \alpha_A} < 1 . 
\end{equation}
Moreover, recall that $p_C =0$ when $X=0$. In this case, 
\begin{equation}
\begin{aligned}
a_1 &= d_S d_H  d_C + d_A d_C \left( d_S + d_H \right) + d_S d_H d_A \left( 1 - \frac{p_S n_H n_A}{d_A d_S d_H} \right) \\
& > d_S d_H  d_C + d_A d_C \left( d_S + d_H \right) , 
\end{aligned}
\end{equation}
and, $a_1 >0$. 
\end{proof}

\section{Discussion}
\label{sec:discussion}
Flexible and stable control of the cell state according to the environmental conditions is the basis for realizing a robust life system. 
In general, hysteretic control is one of the methods that exhibits a prompt and stable response to environmental changes. 
In this work, we clarified the necessary and sufficient conditions for the hysteretic control of the cell state, considering the case of the bacterial cell type regulation as an example.
By incorporating this cell type regulation model into a model of the cell population dynamics \citep{RN236,RN238,RN258,RN257,RN259,RN260,RN285}, the colony morphology can be predicted under a wide range of environmental conditions. 
In particular, it is possible to correctly reproduce the formation of concentric colonies that expand periodically, which has not been realized so far \citep{RN260}. Moreover, the effect of the environmental conditions can be compared with the experimental findings. 

In addition, we developed a model of the cell type regulation influenced by the environmental pH changes, and the findings were noted to be consistent with those of the concentric colony formation experiment. 
In the existing studies, the concentration of agar in the medium, which is a control parameter of the cell population motility, has been widely examined as an environmental factor. 
It is expected that the growth dynamics of concentric circle colonies depending on the agar concentration can be discussed in combination with the model of the cell population dynamics. 

Furthermore, the structure of the hysteresis for external signals presented herein is not limited to \textit{B. subtilis} cell type selection. 
Specifically, a similar structure occurs for phenomena that can be expressed in the form of model \eqref{eq:SHAC}. Furthermore, similar properties are expected if the regulatory network is similar to that shown in Fig.~\ref{fig:1}A. 
In addition, as indicated previously, the main result, Theorem~{\ref{thm:1}}, can be expressed in terms of only the indices $\alpha_A$ and $\alpha_{SH}$, which represent the activation rate parameters from outside the system, as in Corollaries~{\ref{cor:1}} and {\ref{cor:2}}, respectively. 
Such expressions are universal and can help elucidate the regulatory mechanisms of cell populations. 
This analysis methodology, which focuses on the presence or absence of hysteresis related to external signals, can provide a basis for the comprehensive understanding of other control systems, among other applications.

\backmatter

%
%
%

\bmhead{Acknowledgements}

This work was supported by JSPS KAKENHI, Grant Number 19K03645 (S.T.), 23K03208 (S.T.), 23K03225 (M.N.), 23K03176 (I.T.), and MEXT KAKENHI, Grant Number 17H06327 (S.T.).

\section*{Declarations}

\bmhead{Conflict of interest}

The authors have no conflict of interest to declare.


\bibliographystyle{sn-vancouver} 

\end{document}